\RequirePackage{fix-cm}
\documentclass[smallcondensed]{svjour3}
\usepackage[utf8x]{inputenc}
\usepackage[english]{babel}
\usepackage{amsmath,amsfonts,amssymb}
\usepackage{mathrsfs}
\usepackage{latexsym,mathptmx}
\usepackage{epsfig,graphicx,color}
\usepackage{verbatim,listings}
\usepackage{mdwlist}
\usepackage{bm}
\usepackage{dsfont}
\usepackage{oldgerm}
\usepackage{relsize}
\usepackage{afterpage}
\usepackage{lmodern}
\usepackage{tabularx}
\usepackage{multicol}
\usepackage{cancel}
\newcommand{\dd}{\mathrm{d}}

\newcommand{\trans}{^\mathsf{T}}
\newcommand{\av}{\bm{a}}
\newcommand{\n}{\bm{l}_1}
\newcommand{\m}{\bm{r}_1}
\newcommand{\nper}{\bm{l}_2}
\newcommand{\mper}{\bm{r}_2}
\newcommand{\cv}{\bm{c}}
\newcommand{\ca}{\bm{c}^\ast}

\newcommand{\da}{\bm{d}^\ast}
\newcommand{\e}{\bm{e}}
\newcommand{\x}{\bm{x}}
\newcommand{\zero}{\bm{0}}
\newcommand{\y}{\bm{y}}
\newcommand{\nablatwo}{\nabla^2}
\newcommand{\normal}{\bm{\nu}}
\newcommand{\surface}{\mathscr{S}}

\newcommand{\nablas}{\nabla\!_\mathrm{s}}

\newcommand{\f}{\bm{f}}

\newcommand{\vv}{\bm{v}}

\newcommand{\F}{\mathbf{F}}
\newcommand{\I}{\mathbf{I}}

\newcommand{\real}{\mathbb{R}}
\newcommand{\euclid}{\mathscr{E}}
\newcommand{\transla}{\mathscr{V}}

\newcommand{\C}{\mathbf{C}}
\newcommand{\B}{\mathbf{B}}

\newcommand{\curl}{\operatorname{curl}}
\newcommand{\tr}{\operatorname{tr}}

\newcommand{\diam}{\operatorname{diam}}
\newcommand{\E}{\mathbf{E}}
\newcommand{\Ef}{\mathbf{E}_{\f}}
\newcommand{\Cf}{\C_{\f}}
\newcommand{\Cp}{\C_\phi}
\newcommand{\nay}{(\nabla\y)}
\newcommand{\naty}{(\nablatwo\y)}
\newcommand{\nan}{(\nabla\normal)}
\newcommand{\curvature}{(\nablas\normal)}
\newcommand{\A}{\mathbf{A}}
\newcommand{\slab}{\mathsf{S}}
\newcommand{\dir}{\bm{d}}
\newcommand{\framem}{(\m,\mper,\e_3)}
\newcommand{\framen}{(\n,\nper,\normal)}

\smartqed
\title{On the Kirchhoff-Love hypothesis\\  (revised and vindicated)}
\author{Olivier Ozenda and Epifanio G. Virga}
\institute{O. Ozenda \at Dipartimento di Matematica, Universit\`a di Pavia, Via Ferrata 5, 27100 Pavia, Italy \\
            \email{o.ozenda@unipv.it}           
           \and
           E.G. Virga \at Dipartimento di Matematica, Universit\`a di Pavia, Via Ferrata 5, 27100 Pavia, Italy\\
           \email{eg.virga@unipv.it}    
}
\date{\today}
\begin{document}
\maketitle
\begin{abstract}
The Kirchhoff-Love hypothesis expresses a kinematic constraint that is assumed to be valid for the deformations of a three-dimensional body when one of its dimensions is much smaller than the other two, as is the case for plates. This hypothesis has a long history checkered with the vicissitudes of life: even its attribution has been questioned, and recent rigorous dimension-reduction tools (based on $\Gamma$-convergence) have proven to be incompatible with it. We find that an appropriately revised version of the Kirchhoff-Love hypothesis is a valuable means to derive a two-dimensional variational model for elastic plates from a three-dimensional nonlinear free-energy functional. The bending energies thus obtained for a number of materials also show to contain measures of stretching of the plate's mid surface (alongside the expected measures of bending). The incompatibility with $\Gamma$-convergence also appears to be removed in the cases where contact with that method and ours can be made.
\end{abstract}
\section{Introduction}\label{sec:intro}
In a postmodern view, the theory of elasticity is a dead subject. By contrast, we rather hold that it is just deceptively simple: it makes one believe that everything is understood and  only routine computations need to be done, for which it suffices to devise the most appropriate algorithm (the distinguished job of computational mechanics). Nothing farthest from truth.

Perhaps a tangible sign of this is the revival of interest that new problems, mostly arising from soft matter physics, have engendered. Among these problems is that of  predicting the shape that a nematic polymeric network can take upon the action of external stimuli that affect its internal material organization. Nematic polymeric networks are elastomeric materials that bear elongated molecules attached  to the structural backbone, molecules capable of becoming ordered in orientation, as is typical of nematic liquid crystals. Light and heat can interfere with the orientational order of the nematic component, and this in turn can affect the polymeric backbone, inducing stresses that alter the shape of the body. Neat and rich reviews of the many interesting phenomena displayed by nematic elastomers can be found in \cite{warner:topographic} and \cite{white:programmable}, which we highly recommend reading.

Shape is the main object then; especially, when the body is a thin sheet, and so it is more prone to exhibit extraordinary changes of shape as a result of tenuous stimuli. Perhaps, the first stunning manifestation of the potential for applications hidden here was the swimmer ``that swims into the dark'' \cite{camacho-lopez:fast}, a sheet flapping on a fluid surface whenever reached by light. Other soft matter mechanics papers that draw on shape have a biological inspiration, such as \cite{gladman:biomimetic} and \cite{siefert:bio-inspired}. 

The statics of nematic elastomers in three space dimensions is described by an elastic free-energy density, commonly delivered by the ``trace formula'' of Warner and Terentjev's theory \cite{warner:liquid}. In essence, this theory extends the ideas underlying the isotropic Gaussian distribution of polymeric chains to chains made anisotropic by the mutual interactions of the nematogenic molecules appended to them. It is no surprise then if the free energy of nematic elastomers (of a purely entropic nature) turns out to be an anisotropic extension of the classical neo-Hookean formula of isotropic rubber elasticity. 

This formula is valid in the bulk, but we are interested in thin sheets. If a heuristic stretching energy is rather easy to obtain from the trace formula, and lately it has widely been used \cite{modes:disclination,plucinsky:programming,modes:gaussian,modes:negative,mostajeran:curvature,mostajeran:encoding,mostajeran:frame,kowalski:curvature,warner:nematic}, an attempt at deriving an appropriate bending energy is only very recent \cite{ozenda:blend}. A revised Kirchhoff-Love hypothesis, central to the theory of elastic plates and shells, proved  particularly instrumental to our derivation. Here, we wish to revisit this classical hypothesis in its natural environment, as it were, and show how its revision could possibly be used to derive reduced theories for nonlinear elastic plates in which stretching and bending energies are naturally combined together. 

Saying something really new on this subject is very hard. The early theories of plates derived from three-dimensional elasticity for linearly elastic materials were already proposed by Poisson~\cite{poisson:memoire} and Cauchy~\cite{cauchy:exercises}, building on formal expansions of strains and stresses in powers of the transverse coordinate $x_3$, ranging through the interval $[-h,h]$ representing the plate's thickness. Similar expansions, perhaps closer in spirit to the approach that we shall take, were considered by L\'evy~\cite{levy:memoire} in a paper forgotten by most, but revived in \cite{reissner:reflections}. Undoubtedly, the theory of plates was first established on firm grounds by Kirchhoff~\cite{kirchhoff:uber}, who later in his book \cite{kirchhoff:mechanik}, improving on Gehring's dissertation (published in Berlin in 1860), established the elastic energy of a plate as an expression that
\begin{quote}
	``consists of two parts: one a quadratic function of the quantities defining the extension of the middle-surface with a coefficient proportional to the thickness of the plate, and the other a quadratic function of the quantities defining the flexure of the middle-surface with a coefficient proportional to the cube of the thickness.'' \cite[p.\,27]{love:treatise}
\end{quote} 

Kirchhoff's plate theory is based on a number of hypotheses, which we found lucidly described in the account given in \cite{podio:exact}, based on a critical comparison of the statements appearing in revered books, such as \cite{timoshenko:theory} and \cite{novozhilov:foundations}. In essence, Kirchhoff's assumptions can be reduced to two kinematic statements \cite{podio:exact}:
\begin{enumerate} 
	\item Points of the plate lying initially on a normal to the middle plane of the plate remain on the normal to the middle surface of the plate after bending.
\item The distance of every point of the plate from the middle surface remains unchanged by the deformation.
\end{enumerate}
For mysterious reasons, these statements have become to be known as the \emph{Kirchhoff-Love hypothesis}; we stick to this not fully justified tradition.

This hypothesis has been variously criticized in the literature, mainly for the inconsistencies that it may cause with the distribution of stresses that (depending on the specific constitutive law) should sustain the features assumed for the deformation. Podio-Guidugli~\cite{podio:exact,podio:constraint} overcame this criticism by taking the view that the Kirchhoff-Love hypothesis is a constraint ob the admissible deformations, which is sustained by an appropriate reactive stress, specified in a class admitted by symmetry and determined from the equilibrium equations of three-dimensional elasticity.

More recently, the rigorous analytical tool of $\Gamma$-convergence has been employed to derive plate (and shell) theories from three-dimensional elasticity. Admittedly, the problem with this method is that presently it only affords to derive single powers of the energy expansion in the thickness $2h$. Said differently, we can obtain from a three-dimensional constitutive model either a ``membrane-dominated model'' or a ``flexural-dominated model'' (in the words of \cite{CiaMad-2018}), meaning that we can isolate two-dimensional energies either linear or cubic in $h$, respectively. For example, models in the former category have been derived in \cite{bourquin:gamma} for linear plates and in \cite{ledret:nonlinear_membrane} for nonlinear ones (as well as in \cite{ledret:membrane_shell} for nonlinear shells). Models in the latter category have been derived in \cite{friesecke:theorem,friesecke:hierarchy} for nonlinear plates (as well as in \cite{friesecke:derivation_shell} for nonlinear shells). These higher-order $\Gamma$-limits, however, need to be evaluated on the class of deformations that minimize the lower order. In other words, we may only recover the $h$-cubic bending energy on the minimizers of the $h$-linear stretching energy.\footnote{More specifically, in \cite{friesecke:theorem} the bending energy is obtained by computing the $\Gamma$-limit of the energy of nonlinear elasticity over the isometric embeddings of the plate's mid plane.} Sadly, the notion of  $\Gamma$-limit has not yet evolved into that of  $\Gamma$-expansion, and thus it does not yet serve the purpose of deriving blended stretching and bending energies, free to conspire together in a thin sheet of a (possible activable) elastic material, which is our objective here.

Limited as the derivation of the bending energy via $\Gamma$-convergence may be, it raised further critiques against the Kirchhoff-Love hypothesis. It was proven in \cite{friesecke:theorem} that the rigorous bending energy (in the flexural-dominated model) is incompatible with the deformation field assumed by the Kirchhoff-Love hypothesis. Here, we try and remedy this shortcoming by revising (and salvaging) the classical hypothesis. In particular, we shall see that the incompatibility pointed out in \cite{friesecke:theorem} is resolved by our revised hypothesis.

The paper is organized as follows. In Sect.~\ref{sec-mod}, we recall the basic kinematics of plates and present our revision of the classical Kirchhoff-Love hypothesis. In Sect.~\ref{sec:connectors}, we introduce a description for the deformation of a smooth surface that relies on a notion of \emph{Cartesian connectors}, which avoid the use of coordinates and Christoffel symbols. Sections \ref{sec-gent} and \ref{sec-compres} are devoted to the dimension-reduction afforded by our revised kinematic hypothesis in two distinct constitutive classes, one for incompressible plates ad the other for compressible ones. We consider a number of special nonlinear elastic models for the application of our method; for all we derive stretching and bending energies. The features that these results have in common is the presence in the bending energy of stretching measures of the mid surface of the plate  (alongside the expected measures of bending). In Sect.~\ref{sec:conclusion}, we summarize our conclusions and comment on some possible avenues along which this work could be extended. This paper is closed by an Appendix, where we give explicit formulae for the mean and Gaussian curvatures of the deformed mid surface in terms of the mapping that describes it.

\section{Kinematics of plates}
\label{sec-mod}
Here we wish to describe the deformation of an elastic plate  with a uniform width
(and a planar reference configuration). The deformation  will be split into two components,
a planar one, which maps the reference mid plane surface  onto a mid deformed surface, and an axial one, which maps 
vectors normal  to the mid reference surface on  vectors normal to the deformed mid surface.
The classical \emph{Kirchhoff-Love hypothesis} consists in assuming that the second mapping is an isometry (see, for example, \cite[p.\,551]{villaggio:mathematical} and \cite[p.\,156]{ciarlet:introduction}).
Because of this isometry, assuming regularity for the planar mapping is enough to ensure an admissible deformation
for sufficiently thin plates. In the following, the latter assumption will be made precise and the isometric constraint along normals will be relaxed. In this framework, an approximate right Cauchy-Green 
tensor will be constructed and its invariants computed.

\subsection{Kinematic preliminaries}
Let $S$ be a bounded, two-dimensional flat domain immersed in three-dimensional Euclidean space $\euclid$ and  $h>0$ a real constant. We call
$\y :S\rightarrow\euclid$ an injective $\mathcal{C}^3$-immersion of $S$ and we denote its image as $\surface=\y(S)$. Pursuing our aim of 
extending the Kirchhoff-Love hypothesis, we
interpret the closed set $\slab:=\overline{S}\times[-h,h]\subset\euclid$ as the   reference configuration of an elastic plate
whose mid surface is $\overline{S}$. As we focus on the case $2h\ll\diam S$, 
we set $\diam S=1$, for simplicity, meaning that we shall rescale all lengths to $\diam S$.

We define the mapping $\f :\slab\rightarrow\euclid$ as
\begin{equation}
\label{eq:f_definition}
\f(\x,x_3)=\y(\x)+\phi(\x,x_3)\normal(\x),
\end{equation}
where $\normal$ is the unit normal vector
to $\surface$  and $\phi :\slab\rightarrow\real$ is a $\mathcal C^2$-function which 
describes how normals to $S$ deform into normals to  $\surface$. It follows from \eqref{eq:f_definition} that the deformation gradient reads as
\begin{equation}\label{eq:deformation_gradient}
\F(\x,x_3):=\mathrm{D}\f=\nabla\y+\phi\nabla\normal+\phi'\normal\otimes\e_3+\normal\otimes\nabla\phi,
\end{equation}
where $\nabla$ denotes the gradient in $\x$, a prime $'$ denotes differentiation with respect to $x_3$, and $\e_3$ is the unit normal to $S$.\footnote{Here we identify the set $S$ with its trivial embedding in $\euclid$.}
Furthermore, $\phi$ is assumed to obey
\begin{equation}
\label{eq-orient}
   \left\{\begin{array}{l}
          \phi(\x,0)=0,\\
          \phi'(\x,0)>0,
   \end{array}\right.
   \qquad\forall\, \x\in \overline{S},
\end{equation}
which is justified by the requirement that $\f$ be orientation-preserving at least for $x_3=0$, as there, by \eqref{eq:deformation_gradient},
\begin{equation}
\label{eq:det_F}
\det\F(\x,0)=\phi'(\x,0)\det\nay>0.
\end{equation}

The classical Kirchhoff-Love hypothesis just requires that $\phi\equiv x_3$, and so it trivially complies with \eqref{eq-orient}. In our approach, $\phi$ will rather remain free and either determined to enforce the constraint of bulk incompressibility or used to minimize the elastic energy stored across the thickness of the deformed plate. We see now how \eqref{eq-orient} can ensure that
$\f$ is an orientation-preserving $\mathcal C^2$-diffeomorphism onto its image, for  appropriately small values of $h$. 
\begin{proposition}
\label{prop-prelim}
There is $h>0$ such that the mapping $\f$ defined in \eqref{eq:f_definition} is a $\mathcal C^2$-diffeomorphism
from $\slab$ onto $\f(\slab)$ and 
 $\det\F(\x,x_3)>0\ \forall\, (\x,x_3)\in\slab$.
\end{proposition}
\begin{proof}
In the following two steps, we adapt  Theorem 4.1-1 of \cite[p.\,157]{ciarlet:introduction} to our setting.
\begin{enumerate}
\item By the continuity of $\F$, inequality \eqref{eq:det_F} implies that 
\begin{equation*}
\exists\, h_1>0\ \text{such that}\ \det\F(\x,x_3)>0\ \forall\, (\x,x_3)\in \overline{S}\times[-h_1,h_1].
\end{equation*}
\item
As $\phi'(\x,0)>0\ \forall\, \x\in \overline{S}$, the implicit function theorem can be applied on the 
closure  $\overline{S}\times[-h_1,h_1]$. Then, there are $\eta(\x)\in(0,h_1)$ and a neighborhood $U(\x)$ of $\x$ in $\overline{S}$ such that 
$\f$ is a $\mathcal C^2$-diffeomorphism from $U(\x)\times[-\eta(\x),\eta(\x)]$ onto $\f(U(\x)\times[-\eta(\x),\eta(\x)])$. Since
$\overline{S}$ is compact, $\eta$ attains it minimum in $\overline{S}$. Moreover, the minimum of $\eta$ over $\overline{S}$ must be strictly positive, otherwise  $\y$ would fail to be an injective immersion. 
Thus, $h$ can be chosen so that 
\begin{equation}
\label{eq:h_definition}
0<h<\min_{\x\in\overline{S}}(\eta(\x)),
\end{equation}
which is where  we shall hereafter take it to be.
\end{enumerate}
\qed
\end{proof}
In Sects. \ref{sec-gent} and \ref{sec-compres}, we shall use a polynomial approximation for $\phi$ in computing the invariants of the right Cauchy-Green tensor $\Cf$ associated with $\f$. Now, we justify this approximation and lay down a number of preliminary formulae for the invariants of $\Cf$.

\subsection{Invariants of $\Cf$}
\label{sub-approx}
We learned in Prop.~\ref{prop-prelim} how to choose $h>0$ sufficiently small so that the mapping $\f$ is a  $\mathcal C^2$-diffeomorphism. Hypothesis \eqref{eq-orient} also implies that  $\nabla\phi(\x,0)=0\ \forall\, \x\in\overline{S}$; thus, it is also possible
to choose $h$ so small that 
\begin{equation}
\label{eq:inequality}
|\nabla\phi(\x,x_3)|\ll\phi'(\x,x_3)\quad\forall\, (\x,x_3)\in\slab.
\end{equation}
This inequality will be assumed to be valid in the following, and $h$ will be taken to comply with both \eqref{eq:h_definition} and \eqref{eq:inequality}.  Within the approximation stated in \eqref{eq:inequality}, $\F$ in \eqref{eq:deformation_gradient} will be written as 
\begin{equation}
\label{eq:deformation_gradient_rewritten}
\F=\nabla\y+\phi\nabla\normal+\phi'\normal\otimes\e_3.
\end{equation}
\begin{definition}
\label{def-cauchy}
The corresponding \emph{right} Cauchy-Green tensor $\Cf$ associated with $\f$ is given by
\begin{subequations}\label{eq:definition_1}
\begin{equation}\label{eq:C_f_definition}
 \Cf:=\F\trans \F=\C_\phi+\phi'^2\e_3\otimes\e_3,
\end{equation}
where
\begin{equation}
\label{eq:C_phi_definition}
 \C_\phi:=\C+\phi\C_1+\phi^2\C_2
\end{equation}
and
\begin{equation}
\label{eq:C_C1_C2}
\C:=\nay\trans\nay,\quad\C_1:=\nay\trans\nan+\nan\trans\nay,\quad\C_2:=\nan\trans\nan.
\end{equation}
\end{subequations}
\end{definition}
Here $\C$ is the \emph{right} Cauchy-Green tensor associated with the deformation $\y$, while
\begin{equation}
\label{eq:B_definition}
\B:=\nay\nay\trans
\end{equation}
is the \emph{left} Cauchy-Green tensor associated with the same deformation.
The reader should heed that all tensors $\C$, $\C_1$, and $\C_2$ act on the two-dimensional space $\transla_3:=\{\vv\in\transla:\vv\cdot\e_3=0\}$, where $\transla$ is the translation space associated with three-dimensional Euclidean space $\euclid$. $\B(\x)$, however, at the place $\y(\x)\in\surface$,  acts on the two-dimensional space $\transla_{\normal}:=\{\vv\in\transla:\vv\cdot\normal=0\}$.
\begin{remark}
\label{rem-I1}
The curvature tensor $\nablas\normal$ of $\surface$, where $\nablas$ denotes the surface gradient on $\surface$, is a symmetric tensor on $\transla_{\normal}$ (see, for example, \cite{gurtin:continuum}). We easily see that both tensors $\C_1$ and $\C_2$ can be expressed in terms of $\nablas\normal$. As $\nablas\normal=\nan\nay^{-1}$, it readily follows from \eqref{eq:C_C1_C2} and the symmetry of $\nablas\normal$ that 
\begin{equation}\label{eq:C1_C2_curvature}
\C_1=2\nay\trans\curvature\nay,
\qquad
\C_2=\nay\trans\curvature^2\nay.
\end{equation}
\end{remark}

We now compute the principal invariants of $\Cf$.
\begin{proposition}\label{prop:I_1}
The first invariant, $I_1=\tr\Cf$, can be given the form
\begin{equation}\label{eq:I_1}
I_1=(1-\phi^2K)\tr\C+2\phi(1+\phi H)\tr(\B\nablas\normal)+\phi'^2,
\end{equation}
where
\begin{equation}
\label{eq:curvatures}
H:=\frac 12\tr\curvature\quad\text{and}\quad K:=\det\curvature
\end{equation}
are the \emph{mean} and  \emph{Gaussian} curvatures of $\surface$, respectively.
\end{proposition}
\begin{proof}
	We see   from equations \eqref{eq:definition_1} and Remark~\ref{rem-I1} that
\begin{equation}\label{eq:pre_I_1}
I_1=\tr\C+2\phi\tr(\B\curvature)+\phi^2\tr(\B\curvature^2)+\phi'^2.
\end{equation}
The desired conclusion follows from the identity,
\begin{equation}
\label{eq:identity}
\tr(\B\curvature^2)=2H\tr(\B\nablas\normal)-K\tr\C,
\end{equation}
which we now proceed to prove. First, we represent locally the curvature tensor $\nablas\normal$ of $\surface$ as
\begin{equation}
\label{eq:curvature_representation}
\nablas\normal=\kappa_1\bm{n}_1\otimes\bm{n}_1+\kappa_2\bm{n}_2\otimes\bm{n}_2,
\end{equation}
where $\kappa_1$, $\kappa_2$ are the principal curvatures of $\surface$ and $\bm{n}_1$, $\bm{n}_2$, orthogonal unit vectors of $\transla_{\normal}$, are the corresponding principal directions of curvature, so that,  at each place on $\surface$, $\nablas\normal$ is a symmetric tensor acting on  $\transla_{\normal}$. Since $\B$ is also a symmetric tensor acting on $\transla_{\normal}$, it can be represented in the frame $(\bm{n}_1,\bm{n}_2)$ as
\begin{equation}\label{eq:B_representation}
\B=B_{11}\bm{n}_1\otimes\bm{n}_1+B_{22}\bm{n}_2\otimes\bm{n}_2+B_{12}(\bm{n}_1\otimes\bm{n}_2+\bm{n}_2\otimes\bm{n}_1).
\end{equation} 
Then,
\begin{equation}
\label{eq:intermediate_step}
\tr(\B\curvature^2)-2H\tr(\B\nablas\normal)=B_{11}\kappa_1^2+B_{22}\kappa_2^2-(\kappa_1+\kappa_2)(B_{11}\kappa_1+B_{22}\kappa_2)=-K\tr\B,
\end{equation}	
which is precisely \eqref{eq:identity}, as by \eqref{eq:B_definition} $\tr\B=\tr\C$. Use of \eqref{eq:identity} in \eqref{eq:pre_I_1} finally proves \eqref{eq:I_1}. 
	\qed
\end{proof}		
Equation \eqref{eq:I_1} shows that 
$I_1$ involves at most  quadratic terms in $\phi$ and $\phi'$.  As a consequence of the orthogonal decomposition in \eqref{eq:C_f_definition}, the second and third invariants of $\Cf$, $I_2$ and $I_3$,  will also involve higher powers of $\phi$, but not of $\phi'$. To justify the power expansion of the stored elastic energy considered in the following, we need only retain in $I_2$ and $I_3$ the terms at most quadratic in $\phi$; all higher powers of $\phi$ will be neglected. 
\begin{proposition}
\label{prop-I3}
The third invariant, $I_3=\det\Cf$, is expressed by
\begin{equation}
\label{eq:I_3}
I_3=\phi'^2\det\C\left(1+4H\phi+(4H^2+2K)\phi^2\right)+O\left(\phi^3\right).
\end{equation}
\end{proposition}
\begin{proof}
First, we note  that, by \eqref{eq:C_f_definition},
\begin{equation}\label{eq:pre_I_3}
I_3=\phi'^2\det\Cp.
\end{equation}
Then we consider two elementary identities valid for any second-order tensor $\A$ on $\transla_3$: for any othonormal basis $(\e_1,\e_2)$ of $\transla_3$,
\begin{equation}\label{eq:identities}
\det\A=\A\e_1\times\A\e_2\cdot\e_3\quad\text{and}\quad\tr\A=(\A\e_1\times\e_2+\e_1\times\A\e_2)\cdot\e_3.
\end{equation}
Making use of these identities, we readily obtain from \eqref{eq:C_phi_definition} that
\begin{equation}\label{eq:det_C_phi}
\det\C_\phi=\det\C\left(1+\phi\tr\left(\C^{-1}\C_1\right)+\phi^2\tr\left(\C^{-1}\C_2\right)\right)
             +\phi^2\det\C_1 +O\left(\phi^3\right)
\end{equation}
Basic properties of trace and determinant ensure that
\begin{subequations}\label{eq:basis_properties}
\begin{eqnarray}
\tr\left(\C^{-1}\C_1\right)&=&2\tr\curvature=4H,
\\
\tr\left(\C^{-1}\C_2\right)&=&\tr\curvature^2,
\\
\det\C_1&=&4K\det\B=4K\det\C.
\end{eqnarray}
\end{subequations}
Moreover, by the Cayley-Hamilton theorem,
\begin{equation*}
\tr\curvature^2=\tr^2\curvature-2\det\curvature=4H^2-2K,
\end{equation*}
which together with \eqref{eq:pre_I_3}, \eqref{eq:det_C_phi}, and \eqref{eq:basis_properties} lead us to \eqref{eq:I_3}.
\qed
\end{proof}
\begin{proposition}\label{prop:I_2}
The second invariant $I_2$ of $\Cf$ is given by
\begin{eqnarray}\label{eq:I_2}
I_2:&=&\frac12\left(\tr^2\Cf-\tr\Cf^2\right)\nonumber
\\
&=&\det\C\left(1+4H\phi+(4H^2+2K)\phi^2\right)\nonumber
 \\
   &+&\phi'^2\left((1-\phi^2K)\tr\C+2\phi(1+\phi H)\tr(\B\nablas\normal)\right)
   +O\left(\phi^3\right).
\end{eqnarray}
\end{proposition}
\begin{proof}
	It follows from \eqref{eq:C_f_definition} that 
	\begin{equation}\label{eq:I_2_C_phi}
	I_2=\frac 12 \left(\tr^2\C_\phi-\tr\C_\phi^2\right)+\phi'^2\tr\C_\phi.
	\end{equation}
Since $\Cp$ is tensor on $\transla_3$, again by the Cayley-Hamilton theorem, it satisfies
\begin{equation}
\label{eq:C__H_for _C_phi}
\Cp^2-(\tr\Cp)\Cp+(\det\Cp)\mathbf{I}_2=\bm{0},
\end{equation} 
where $\mathbf{I}_2$ is the identity on $\transla_3$. Taking the trace of both sides of \eqref{eq:C__H_for _C_phi}, we obtain from \eqref{eq:I_2_C_phi} that
\begin{equation}
\label{eq:pre_I_2}
I_2=\det\Cp+\phi'^2\tr\Cp.
\end{equation}
The desired conclusion then follows from \eqref{eq:det_C_phi}, \eqref{eq:basis_properties}, and \eqref{eq:I_1}, since $\tr\Cp=I_1-\phi'^2$.
\qed
\end{proof}
In Sect.~\ref{sec-gent}, we shall consider materials that obey the incompressibility constraint, $I_3=1$. There, \eqref{eq:I_3} will turn into a differential equation that determines $\phi$.

In preparation for this, in the following section we refresh the preliminaries of differential geometry of surfaces in a way that avoids local charts of coordinates, but resorts instead to a number of vector fields, which describe the correspondence between local movable frames in the reference and current configurations of a material surface.

\section{Cartesian connectors}\label{sec:connectors}
Here, we introduce the notion of \emph{Cartesian connectors}, which in our view constitute a viable alternative to Christoffel symbols. In terms of these connectors, we reformulate the classical \emph{theorema egregium} of Gauss and the Codazzi-Mainardi compatibility conditions. 

We reformulate the essentials of the differential geometry of smooth surfaces embedded in three-dimensional space. For definiteness, we shall assume that the mapping $\y$ that deforms $S$ into $\surface$ is of class $\mathcal{C}^3$. 

Letting $(\m,\mper)$ be the \emph{right} principal directions, that is , the (normalized) eigenvectors of $\C$, and $(\n,\nper)$ the \emph{left} principal directions, that is, the (normalized) eigenvectors of $\B$, with corresponding principal stretches (common to both tensors) $\lambda_1>0$ and $\lambda_2>0$, we may represent $\nabla\y$, $\C$, and $\B$ as follows (see, for example, \cite[p.\,74]{gurtin:mechanics}),
\begin{subequations}\label{eq:representations}
\begin{align}
\nabla\y&=\lambda_1\n\otimes\m+\lambda_2\nper\otimes\mper,\label{eq:representation_nabla_y}\\
\C&=\lambda_1^2\m\otimes\m+\lambda_2^2\mper\otimes\mper,\label{eq:representation_C}\\
\B&=\lambda_1^2\n\otimes\n+\lambda_2^2\nper\otimes\nper.\label{eq:representation_B}
\end{align}	
\end{subequations}
We shall assume that the frames $\framem$ and $\framen$ are oriented so that $\e_3=\m\times\mper$ and $\normal=\n\times\nper$. It should be kept in mind that both $\m$ and $\mper$ lie in the $(x_1,x_2)$ plane; $\nabla$ denotes the two-dimensional gradient in this plane, whereas $\nablas$  denotes the surface gradient on $\surface$.

The \emph{connector} $\cv$ is a vector field in the plane such that
\begin{subequations}\label{eq:connector_c}
	\begin{align}
	\nabla\m&=\mper\otimes\cv,\label{eq:connector_c_1}\\
	\nabla\mper&=-\m\otimes\cv.\label{eq:connector_c_2}
	\end{align}
\end{subequations}
The existence of $\cv$ and the specific form of \eqref{eq:connector_c} follow from the requirement that the right principal directions  $(\m,\mper)$ be orthonormal everywhere on $S$.\footnote{The name \emph{connector} is inspired by the notion of \emph{spin connection} for surfaces (and manifolds) (see \cite{kamien:geometry}, for an effective introduction to the differential geometry useful in modelling soft matter). Here, we have applied \eqref{eq:connector_c} to the right principal directions; clearly, the same equations apply to any other pair of orthogonal directions.} Clearly, if $\m$ is known then $\cv$ is defined as $\cv:=(\nabla\m)\trans\mper$; on the other hand, if $\cv$ is assigned, at least locally, in the class $\mathcal{C}^1$ then $\m$ (and $\mper$) can be determined up to a rigid rotation by solving equations \eqref{eq:connector_c}. To this end, however, $\cv$ must be \emph{compatible}; it follows from the symmetry of both $\nablatwo\m$ and $\nablatwo\mper$ that the compatibility condition reads as
\begin{equation}\label{eq:symmetry_grad_c}
\nabla\cv-(\nabla\cv)\trans=\zero,
\end{equation}
which for a simply connected $S$ implies that $\cv=\nabla\Phi$, where $\Phi$ is an appropriate scalar potential. Since we assume that both $\m$ and $\mper$ are determined by $\y$, we shall here consider  $\cv$ as known and satisfying \eqref{eq:symmetry_grad_c}.

In complete analogy to the frame $\framem$ on $S$, 
we  describe the corresponding frame $\framen$ as a field of orthonormal directors on $\surface$. Equations \eqref{eq:connector_c} are generalized to
\begin{subequations}\label{eq:connectors_c_d}
	\begin{align}
	\nabla\n&=\nper\otimes\ca+\normal\otimes\da_1,\label{eq:nabla_n}\\
	\nabla\nper&=-\n\otimes\ca+\normal\otimes\da_2,\label{eq:nabla_n_perp}\\
	\nabla\normal&=-\n\otimes\da_1-\nper\otimes\da_2\label{eq:nabla_nu},
	\end{align}
\end{subequations}
where the connectors $\ca$, $\da_1$, and $\da_2$ are planar fields defined on $S$. A number of consequences for these fields follow from the integrability  condition that requires the second gradients of $\y$, $\n$, $\nper$, and $\normal$ to be symmetric: they are listed below.
\begin{enumerate}
	\item For the symmetry of $\nablatwo\y$ (in its last two legs), the following second-order tensors must be symmetric,
\begin{subequations}\label{eq:pre_symmetry_nabla2_y}
\begin{align}
\n\cdot\naty&=\m\otimes\nabla\lambda_1+\lambda_1\mper\otimes\cv-\lambda_2\mper\otimes\ca,\label{eq:symmetry_nabla2_y_1}\\
\nper\cdot\naty&=\lambda_1\m\otimes\ca+\mper\otimes\nabla\lambda_2-\lambda_2\m\otimes\cv,\label{eq:symmetry_nabla2_y_2}\\
\normal\cdot\naty&=\lambda_1\m\otimes\da_1+\lambda_2\mper\otimes\da_2.\label{eq:symmetry_nabla2_y_3}
\end{align}
\end{subequations}
Thus, for the symmetry of $\nablatwo\y$, it must be	
\begin{subequations}
		\begin{align}
		\lambda_1\cv\cdot\m-\nabla\lambda_1\cdot\mper-\lambda_2\ca\cdot\m&=0	,\label{eq:c_star_m_perp}\\
		\lambda_1\ca\cdot\mper-\lambda_2\cv\cdot\mper-\nabla \lambda_2\cdot\m&=0,\label{eq:c_star_m}\\
	\lambda_1\da_1\cdot\mper-\lambda_2\da_2\cdot\m	&=0.\label{eq:d_12_d21}
		\end{align}
	\end{subequations}
	In particular, \eqref{eq:c_star_m_perp} and \eqref{eq:c_star_m} can be combined together to yield 
	\begin{equation}\label{eq:c_star}
	\ca=\frac{1}{\sqrt{\det\C}}\C\cv-\frac{1}{\lambda_2}(\nabla \lambda_1\cdot\mper)\m+\frac{1}{\lambda_1}(\nabla\lambda_2\cdot\m)\mper.
	\end{equation}
	 By recalling \eqref{eq:representation_C},
	it becomes apparent from \eqref{eq:c_star} that $\ca$ is completely determined by $\cv$ and $\C$.
	\item Similarly, for the symmetry of $\nablatwo\n$, it must be
	\begin{subequations}
		\begin{align}
		\nabla\ca-(\nabla\ca)\trans&=\da_1\otimes\da_2-\da_2\otimes\da_1,\label{eq:skew_part_nabla_c_ast}\\
		\nabla\da_1-(\nabla\da_1)\trans&=\da_2\otimes\ca-\ca\otimes\da_2,\label{eq:skew_part_nabla_d_1}
		\end{align}
	\end{subequations}
	which can also be written in the equivalent forms
	\begin{subequations}
		\begin{align}
		\curl\ca&=\da_2\times\da_1,\label{eq:skew_part_nabla_c_ast_equiv}\\
		\curl\da_1&=\ca\times\da_2.\label{eq:skew_part_nabla_d_1_equiv}
		\end{align}
	\end{subequations}
	\item For the symmetry of $\nablatwo\nper$, \eqref{eq:skew_part_nabla_c_ast} is supplemented by
	\begin{equation}\label{eq:skew_part_d_2}
	\nabla\da_2-(\nabla\da_2)\trans=\ca\otimes\da_1-\da_1\otimes\ca,
	\end{equation}
	or its equivalent form
	\begin{equation}\label{eq:skew_part_d_2_equiv}
	\curl\da_2=\da_1\times\ca.
	\end{equation}
	\item Finally, the symmetry of $\nablatwo\normal$ is guaranteed by \eqref{eq:skew_part_nabla_d_1} and \eqref{eq:skew_part_d_2}.
\end{enumerate} 
The connectors $\da_1$ and $\da_2$ can be given a geometric interpretation by computing the curvature tensor $\nablas\normal$ of $\surface$.
It readily follows from \eqref{eq:representation_nabla_y} that 
\begin{equation}\label{eq:nabla_y_inverse}
(\nabla\y)^{-1}=\frac{1}{\lambda_1}\m\otimes\n+\frac{1}{\lambda_2}\mper\otimes\nper.
\end{equation}
Letting
\begin{equation}\label{eq:d_representation}
\da_1=d_{11}\m+d_{12}\mper,\qquad\da_2=d_{21}\m+d_{22}\mper,
\end{equation}
from \eqref{eq:nabla_y_inverse} we arrive at
\begin{equation}\label{eq:curvature_tensor}
\nablas\normal=(\nabla\normal)(\nabla\y)^{-1}
=-\left(\frac{d_{11}}{\lambda_1}\n\otimes\n+\frac{d_{12}}{\lambda_2}\n\otimes\nper+\frac{d_{21}}{\lambda_1}\nper\otimes\n+\frac{d_{22}}{\lambda_2}\nper\otimes\nper\right),
\end{equation}
which is duly symmetric, as by \eqref{eq:d_representation} equation \eqref{eq:d_12_d21} reduces to
\begin{equation}\label{eq:curvature_symmetry}
\lambda_2d_{21}=\lambda_1d_{12}.
\end{equation}
Both the mean curvature $H$ and the Gaussian curvature $K$ of $\surface$ can easily be derived from \eqref{eq:curvature_tensor}; they are given by
\begin{align}
H&=-\frac12\left(\frac{d_{11}}{\lambda_1}+\frac{d_{22}}{\lambda_2}\right),\label{eq:H}\\
K&=\frac{1}{\lambda_1\lambda_2}\left(d_{11}d_{22}-d_{12}d_{21}\right).\label{eq:K}
\end{align}

An important conclusion follows by combining \eqref{eq:skew_part_nabla_c_ast} and \eqref{eq:K} with the aid of \eqref{eq:d_representation}, namely
\begin{equation}\label{eq:nabla_c_star_skew_part}
\nabla\ca-(\nabla\ca)\trans=\lambda_1\lambda_2K(\m\otimes\mper-\mper\otimes\m).
\end{equation}
Since, as shown by \eqref{eq:c_star}, the left-hand side of \eqref{eq:nabla_c_star_skew_part} is determined by $\C$ (alongside its first and second spatial derivatives), so is $K$. In other words, the metric on $\surface$  determines the Gaussian curvature of $\surface$. This is the manifestation in our setting of the celebrated \emph{theorema egregium} of Gauss. Similarly, equations \eqref{eq:skew_part_nabla_d_1} and \eqref{eq:skew_part_d_2} are related to the Codazzi-Mainardi equations (see, for example, \cite[p.\,144]{stoker:differential}).

In the special case where both principal stretches are uniform in space (but not necessarily the principal directions of stretching), equation \eqref{eq:c_star} reduces to
\begin{equation}\label{eq:c_star_reduced}
\ca=\frac{\lambda_1}{\lambda_2}c_1\m+\frac{\lambda_2}{\lambda_1}c_2\mper,
\end{equation}
where $c_1:=\cv\cdot\m$ and $c_2:=\cv\cdot\mper$.
It follows 
from \eqref{eq:symmetry_grad_c}, \eqref{eq:c_star_reduced}, and the identities
\begin{equation}\label{eq:nabla_c_1_2_identities}
\nabla c_1=(\nabla\cv)\m+c_2\cv,\qquad\nabla c_2=(\nabla\cv)\mper-c_1\cv,
\end{equation}
that
\begin{equation}\label{eq:nabla_c_star_skew_part_pre}
\nabla\ca-(\nabla\ca)\trans=\left(\frac{\lambda_1}{\lambda_2}-\frac{\lambda_2}{\lambda_1}\right)(c_2^2-c_1^2+c_{12})(\m\otimes\mper-\mper\otimes\m),
\end{equation}
where we have set $c_{12}:=\m\cdot(\nabla\cv)\mper$.
A comparison with \eqref{eq:nabla_c_star_skew_part} readily helps us to conclude that 
\begin{equation}\label{eq:theorema_egregium_reduced}
K=\left(\frac{1}{\lambda_2^2}-\frac{1}{\lambda_1^2}\right)(c_2^2-c_1^2+c_{12}).
\end{equation}
It is not difficult to check that \eqref{eq:theorema_egregium_reduced} agrees completely with equation (22) of \cite{mostajeran:curvature}, which was deduced with the more traditional use of coordinates and Christoffel symbols.

As appealing as formulae \eqref{eq:H} and \eqref{eq:K} may be, they are not especially expedient to compute $H$ and $K$, for a given deformation $\y$, as the link between the latter and the connectors is rather intricate. In Appendix~\ref{sec:H_K}, we shall give other formulae for $H$ and $K$ valid for area-preserving deformations $\y$; they are more accessible to direct computation and also show the role played by the second gradient $\nablatwo\y$ in determining the principal curvatures of $\surface$.

\section{Incompressible elastomer plates}
\label{sec-gent}
In this section, we consider an incompressible elastomer plate, for which we assume that the mid surface $S$ is inextensible and the whole body $\slab$ is incompressible. That is, we assume that
\begin{equation}
\label{eq:detC=1}
\det\C=1\quad\text{and}\quad\det\Cf=1.
\end{equation}
Here we consider the former constraint as a remnant of the latter, the one that survives when, in the limit as $h\to0$, only stretching energy is associated with the membrane $S$ by an appropriate dimension reduction of the elastic energy stored in the three-dimensional body $\slab$.\footnote{In Remark~\ref{rem:determinant}, we shall see the consequences of relaxing the inextensibility constraint $\det\C=1$.}
  
\subsection{Polynomial approximation}
It is our desire to compute averages of the elastic  energy stored across the (small) thickness of the plate. To this end, it will suffice to represent $\phi$ as a polynomial in $x_3$. Using the expressions for the invariants of $\Cf$ presented in 
Sect.~\ref{sec-mod}, in the following proposition we shall identify this polynomial. 
\begin{proposition}
\label{prop-phip}
If we let
\begin{equation}\label{eq:phi_polynomial}
\phi(\x,x_3)=\alpha(\x)x_3+\beta(\x)x_3^2+\gamma(\x)x_3^3+O\left(x_3^4\right)\quad  \forall\, (\x,x_3)\in\slab,
\end{equation}
 which complies with \eqref{eq-orient} for $\alpha>0$, then \eqref{eq:detC=1} requires that 
 \begin{equation}\label{eq:alpha_beta_gamma}
    \left\{\begin{array}{l}
      \alpha(\x)=1,\\ 
      \beta(\x)=-H(\x),\\
      \gamma(\x)=\dfrac 13 \left(6H(\x)^2-K(\x)\right) ,
    \end{array}\right.
    \quad\forall\,\x\in\overline{S}. 
 \end{equation}
\end{proposition}
\begin{proof}
By Prop.~\ref{prop-I3}, the constraints in \eqref{eq:detC=1} reduce to the equation 
\begin{equation}
\label{eq:phi_constraint}
\phi'^2\left(1+4H\phi+(4H^2+2K)\phi^2\right)=1,
\end{equation}
where $\phi$ is as in \eqref{eq:phi_polynomial}. 
The desired result then follows by identifying in \eqref{eq:phi_constraint} the coefficients of equal powers of $x_3$ up to $x_3^3$, and recalling that $\alpha>0$.
\qed
\end{proof}
\begin{remark}
The asymptotic expansion for  $\phi$ presented in Prop.~\ref{prop-phip} 
is consistent with a $\mathcal C^3$-regularity for $\phi$ in $x_3$.
As this hypothesis requires more regularity that that envisaged in Sect.~\ref{sec-mod}, this is clearly a particuliar case of the framework described there.
\end{remark}
\begin{remark}
The classical Kirchhoff-Love hypothesis, which requires $\phi\equiv x_3$, can thus be envisaged as the lowest approximation to $\phi$ in \eqref{eq:phi_polynomial}. The higher order approximation represented by \eqref{eq:alpha_beta_gamma} entails a local dependence of the thickness $2h^\ast$ of the deformed plate on the invariant measures of curvature for $\surface$. Explicitly, this coupling is given by
\begin{equation}
\label{eq:thickness_deformed_plate}
2h^\ast=\int_{-h}^h|\F\e_3|\dd x_3=\int_{-h}^h\phi'\dd x_3=2h+2h^3(6H^2-K)+O(h^5).
\end{equation}	
\end{remark}	

\subsection{Gent's material}
The elastic energy stored in a plate made of Gent's material is given by 
\begin{equation}
\label{eq-gent}
W_\mathrm{G}:=-\frac{\mu}{2}J_m\ln\left(1-\frac{I_1-3}{J_m}\right),
\end{equation}
where $I_1=\tr\Cf$ is the first invariant of $\Cf$,  and $\mu$ and $J_m$ are positive material constants, which can be identified with
a shear modulus and a stiffening parameter, respectively.
The role of the latter is illuminated by the request that
\begin{equation}
\label{eq:request}
I_1<J_m+3,
\end{equation}
to which $I_1$ must be subjected for $W_\mathrm{G}$ to be meaningful.

Gent's constitutive law \eqref{eq-gent} was first proposed in \cite{Gen-1996}; it represents the simplest mathematical model for rubber elasticity that accounts for the limited extensibility of the polymeric chains constituting these materials. There is a vast literature on microscopic and phenomenological theories for rubber-like materials based on limited molecular extensibility, for which Beatty \cite{beatty:radial} coined the name of \emph{limited} or \emph{restricted} elastic models; we refer the reader to the emphatic review \cite{Hor-2015}, which focuses on Gent's material.

Taking the limit as $J_m\to\infty$ in \eqref{eq-gent}, we give $W_\mathrm{G}$ the form
\begin{equation}
\label{eq:neo_Hookean}
W_\mathrm{nH}:=\frac12\mu(I_1-3),
\end{equation}  
which is the celebrated neo-Hookean stored energy density, a special case of the Mooney-Rivlin formula,
\begin{equation}
\label{eq:Mooney_Rivlin}
W_\mathrm{MR}:=\frac12\mu[\chi(I_1-3)+(1-\chi)(I_2-3)],
\end{equation}
where $0<\chi\leqq1$ is a dimensionless parameter.\footnote{Clearly, \eqref{eq:Mooney_Rivlin} reduces to \eqref{eq:neo_Hookean} for $\chi=1$.} Neither $W_\mathrm{nH}$ nor $W_\mathrm{MR}$ are capable of describing the severe stiffening that occurs even at moderates stretches for soft biological membranes \cite{holzapfel:similarities}, whereas \eqref{eq-gent} is.

Building upon early work \cite{zuniga:constitutive} that had extended the long-standing tradition of statistical theories for ideal molecules constituted by freely joined  rigid links subject to a non-Gaussian distribution for the end-to-end distance,\footnote{This tradition begun with the works of Kuhn~\cite{kuhn:beziehungen} and Kuhn and Gr\"un~\cite{kuhn:beziehungen_a} (a wider selection of early studies can be found in Treloar's book \cite{treloar:non-linear_third}); in particular, the  paper by Wang and Guth~\cite{wang:statistical} was the starting point of \cite{beatty:average}.} Beatty~\cite{beatty:average} motivated a constitutive law for rubber elasticity depending only on $I_1$ and incorporating the stiffening phenomena associated with a limited extensibility of the constituting chains. It was also shown in \cite{horgan:molecular} that \eqref{eq-gent}, which has a genuine phenomenological origin, is a very accurate approximation to Beatty's molecular-based constitutive law; it retraces all qualitative features of the latter and it reproduces its quantitative predictions, with the advantage of being mathematically simpler, even amenable to explicit, closed-form solutions. As shown in \cite{horgan:molecular}, $\mu$ and $J_m$ can also be related to the molecular model by 
\begin{equation}
\label{eq:molecular_estimates}
\mu=nkT\quad\text{and}\quad J_m=3(N-1),
\end{equation}
where $n$ is the number density (per unit volume) of molecular chains, $N$ is th number of links in each chain, $k$ is the Boltzmann constant, and $T$ the absolute temperature. 

Here, we wish to show how stretching and bending energies are blended together in a thin sheet of Gent's rubber-like material complying with \eqref{eq:detC=1} and \eqref{eq:request}. This result will be achieved in Prop.~\ref{prop-main} by integrating $W_\mathrm{G}$ across the thickness of the plate using the polynomial expression for $\phi$ found in Prop.~\ref{prop-phip}.  
\begin{proposition}
	\label{prop-main}
Let $\phi$ be of class $\mathcal{C}^3$ in $x_3$, so that it can be expressed as in Prop.~\ref{prop-phip} and let the constraints \eqref{eq:detC=1} be enforced.
	Then, the following expression is valid for all $\x\in\overline{S}$,
	\begin{subequations}
	\begin{equation}
		w_\mathrm{G}(\x):=\int_{-h}^h W_\mathrm{G}(\x,x_3){\rm d}x_3
		=
		hw_s(\x)+h^3w_b(\x)+O\left(h^5\right),\label{eq:w_G}
		\end{equation}
		where
		\begin{align}
		w_s(\x)&:=-\mu J_m \ln\left(1-\frac{\tr\C(\x)-2}{J_m}\right),\label{eq:w_G_s}
		\\
		w_b(\x)&:= \frac{\mu }{3}J_m
		\left[
		2\left(\frac{\tr\left(\B(\x)\nablas\normal(\x)\right)-2H(\x)}{J_m-\left(\tr\C(\x)-2\right)}\right)^2
		+
		\frac{16H(\x)^2-K(\x)\left(\tr\C(\x)+2\right)}{J_m-\left(\tr\C(\x)-2\right)}
		\right]
		\label{eq:w_G_b}
	\end{align} 
	\end{subequations}
\end{proposition}
\begin{proof}
By inserting \eqref{eq:I_1} into \eqref{eq-gent} and making use of \eqref{eq:phi_polynomial} and \eqref{eq:alpha_beta_gamma}, we expand $W_\mathrm{G}$ in powers of $x_3$ up to $x_3^2$; standard computations lead us to
\begin{equation}\label{eq:pre_w_G}
\begin{split}
	W_\mathrm{G}=\frac{\mu}{2}J_m&\left[-\ln\left(1-\frac{\tr\C-2}{J_m}\right)
	+2\frac{\tr\left(\B\nablas\normal\right)-2H}{J_m-\left(\tr\C-2\right)}x_3
	\right.
	\\
	&\left.+
	\left(
	2\left(\frac{\tr\left(\B\nablas\normal\right)-2H}{J_m-\left(\tr\C-2\right)}\right)^2
	+
	\frac{16H^2-K\left(\tr\C+2\right)}{J_m-\left(\tr\C-2\right)}
	\right) x_3^2
	\right]
	+O\left(x_3^3\right),
	\end{split}
\end{equation}
where the dependence on $\x$ has been omitted to avoid clutter. Integrating this expression for $W_\mathrm{G}$ across the thickness of the plate, we reach our desired conclusion.
\qed	
\end{proof}
Proposition~\ref{prop-main} is the main result of this section.
The quantities $hw_s $ and $h^3w_3$ introduced in \eqref{eq:w_G}  are interpreted as the stretching and bending elastic energy-densities (per unit area) of Gent's plates. Following \cite{efrati:elastic}, we shall call $w_s$ and $w_b$ the stretching and bending \emph{contents} of Gent's elastic energy, respectively. 
\begin{remark}\label{rem:geometric_elastivity}
There are apparent similarities between the expression for the plate's surface energy density arrived at in Prop.~\ref{prop-main} 	and the elastic energy densities posited in geometric elasticity (see, for example, \cite{efrati:elastic,efrati:buckling,efrati:metric,armon:geometry}), but there are also marked differences. Geometric elasticity of plates (and shells) does blend together stretching and bending energies, which scale with different powers of $h$; the former, like $w_s$ in \eqref{eq:w_G_s}, is of a pure metric nature, while the latter is of a pure curvature nature, unlike $w_b$ in \eqref{eq:w_G_b}, where metric and curvature measures are combined together in an invariant way.
\end{remark}
In the vanishing thickness limit, that is, as $h\to0$, if both $w_s$ and $w_b$ stay bounded, the stretching energy prevails over the bending energy and provides the leading deformation mechanism; for $h$ sufficiently small, we may consider the bending energy as a perturbation to the stretching energy, to be minimized, as it were, at a second stage, on the minimizers of the latter.	
\begin{remark}\label{rem:isometries}
Among all tensors $\C$ on $\transla_3$ such that $\det\C=1$, the stretching content $w_s$ in \eqref{eq:w_G_s} attains its minimum at $\C=\mathbf{I}_2$, which is its unique minimizer.	Indeed, letting $\lambda_2=\frac{1}{\lambda_1}$ in \eqref{eq:representation_C}, we can write $w_s$ as
\begin{equation}
\label{eq:w_G_s_reduced}
\widehat{w}_s=-\mu J_m\ln\left(1-\frac{1}{J_m}\left(\lambda_1-\frac{1}{\lambda_1}\right)^2\right),
\end{equation}
which attains its unique minimum for $\lambda_1=1$, where $w_s$ vanishes. Thus, in the absence of obstructive boundary conditions and external forces, the surfaces $\surface$ that minimize $w_s$ are isometric immersions of $S$ in three-dimensional space. By \eqref{eq:theorema_egregium_reduced}, all such surfaces have $K=0$. Further minimizing $w_b$ on these immersions amounts at minimizing
\begin{equation}
\label{eq:w_G_b_reduced}
\widehat{w}_b=\frac{16}{3}\mu H^2,
\end{equation}
which is the form (for $K=0$) of the energy density featuring in Helfrich's functional for flexible vesicles \cite{helfrich:elastic}.
\end{remark}	
The two-step minimization outlined in Remark~\ref{rem:isometries} is clearly highly hypothetical, for at least two reasons. First, boundary conditions and external forces are always present and are responsible for shaping the equilibrium configurations of plates (especially, elastomer plates, which are more responsive to mechanical stimuli). Second, the representations of stretching and bending contents in \eqref{eq:w_G_s_reduced} and \eqref{eq:w_G_b_reduced}  miss the main points of the full-blown representations in \eqref{eq:w_G_s} and \eqref{eq:w_G_b}, that is, that the measures of stretch influences the bending content as well and that both contents are \emph{blended} together in \eqref{eq:w_G} in a way that depends on $h$ and may give rise to interesting instability scenarios driven by the plate's thickness. The limiting forms of the stretching and bending contents such as $\widehat{w}_s$ and $\widehat{w}_b$ remain however indicative and will also be used in the following section to establish contact with the $\Gamma$-convergence branch of literature in this field, where those limit have been rigorously established for a number of models.
\begin{remark}\label{rem:coipling}
When stretching and bending energies compete one against the other for an equilibrium, the way the surface is stretched affects its response to bending. The coupling between the two energies is not only conveyed through $\tr\C$, the sum of the principal stretches, it also involves the relative orientation of the eigenframes of the curvature tensor $\nablas\normal$ and left Cauchy-Green tensor $\B$.

Letting the former be represented as in \eqref{eq:curvature_representation} and the latter as in \eqref{eq:representation_B}, with $\lambda_2=\frac{1}{\lambda_1}$, we can write
\begin{equation}
\begin{split}
\label{eq:eigenframe_coupling}
w_\varphi&:=[\tr(\B\nablas\normal)-2H]^2=\tr^2[(\B-\mathbf{I}_2)\nablas\normal]\\
&=\left[\left(\kappa_1\lambda_1^2+\frac{\kappa_2}{\lambda_1^2}-\kappa_1-\kappa_2\right)\cos^2\varphi+\left(\kappa_2\lambda_1^2+\frac{\kappa_1}{\lambda_1^2}-\kappa_1-\kappa_2\right)\sin^2\varphi\right]^2,
\end{split}
\end{equation}
where $\varphi$ is the angle that $\n$ makes with $\bm{n}_1$. This is the only contribution to $w_b$ that depends on $\varphi$. Even for given $\lambda_1$, $\kappa_1$, and $\kappa_2$, minimizing $w_\varphi$ is not trivial. While $w_\varphi$ is independent of $\varphi$ in the special case that $\lambda_1=1$ or $\kappa_1=\kappa_2$, in general, it has always two stationary points at $\varphi=0$ and $\varphi=\frac\pi2$, which are somehow expected, as there $\B$ and $\nablas\normal$ share the same eigenframe. However, another pairs of stationary points may arise, for which
\begin{equation}
\label{eq:extra_stationary_pair}
\tan^2\varphi=-\frac{\kappa_1\lambda_1^4-(\kappa_1+\kappa_2)\lambda_1^2+\kappa_2}{\kappa_2\lambda_1^4-(\kappa_1+\kappa_2)\lambda_1^2+\kappa_1},
\end{equation}
provided that $\lambda_1$, $\kappa_1$, and $\kappa_2$ make the right-hand side of \eqref{eq:extra_stationary_pair} positive. These extra stationary points, when they exist, make $w_\varphi$ vanish, so that it attains its infimum. This shows that at equilibrium the relative orientation of $\B$ and $\nablas\normal$ may give rise to interesting patterns on $\surface$.
\end{remark}	  
\begin{remark}\label{rem:determinant}
We have assumed at the start of this section that $S$ is inextensible and thus $\C$ is subject to $\det\C=1$. This constraint can be easily relaxed, while still enforcing $\det\Cf=1$. A few changes occur in our analysis, which otherwise proceeds unaltered. We record here these changes for the interested reader. The polynomial representation formula for $\phi$ in Prop.~\ref{prop-phip} becomes
\begin{equation}
\label{eq:polynomial_phi_new_formula}
\phi=\frac{1}{\sqrt{\det\C}}\left(x_3-\frac{H}{\sqrt{\det\C}}x_3^2+\frac13\frac{6H^2-K}{\det\C}x_3^3 \right),
\end{equation}	
while the expressions for $w_s$ and $w_b$ in \eqref{eq:w_G_s} and \eqref{eq:w_G_b} are to be replaced by
\begin{equation}
\label{eq:w_G_s_new}
w_s=-\mu J_m\ln\left(1-\frac{\det\C(\tr\C-3)+1}{J_m\det\C} \right)
\end{equation}
and
\begin{equation}
\label{eq:w_G_b_new}
w_b=\frac13\mu J_m\frac{1}{\det\C}\Bigg[2\left(\frac{\det\C\tr(\B\nablas\normal)-2H}{\det\C(J_m-\tr\C+3)-1}\right)^2+\frac{16H^2-K(\det\C\tr\C+2)}{\det\C(J_m-\tr\C+3)-1}
 \Bigg],
\end{equation}
respectively. It is a simple matter to check that for $\det\C=1$ equations \eqref{eq:polynomial_phi_new_formula}, \eqref{eq:w_G_s_new}, and \eqref{eq:w_G_b_new} reproduce the corresponding formulae derived above.
\end{remark}
The great advantage offered by the incompressibility constraint $\det\Cf=1$  (and amply exploited in this section)is to determine $\phi$ directly on kinematic grounds, as shown in Prop.~\ref{prop-phip}. For compressible materials, this advantage is lost. We need a different criterion to determine $\phi$. In the following section, we shall show that such a criterion can be found in minimizing the elastic energy stored in the plate, for a given deformation $\y$ of the mid surface $S$.

\section{Compressible plates}
\label{sec-compres}
In this section, we  apply  of the method presented in Sect.~\ref{sec-mod} to 
 compressible materials. We shall show how the modified Kirchhoff-Love  hypothesis purported in this paper  actually entails non-trivial normal strains for a compressible plate. Our analysis, which again is not confined to small strains, will conduce to a blending of stretching and bending energies. To ease the comparison between these latter energies and those already proposed in the literature (mostly for small strains), we shall also consider the small-strain limit for both examples we treat in detail below. In one case, we shall derive a Koiter-like potential \cite{koiter:consistent,koiter:nonlinear,koiter:foundations_I,koiter:foundations_II}  for the Ciarlet-Geymonat material \cite{CiaGey-1982}; this potential is also shown to agree with that recently derived in \cite{CiaMad-2018} for the same material. In the other case, we recover  the bending energy derived in \cite{friesecke:theorem} as a rigorous $\Gamma$-limit on isometries for a variant of the Saint-Venant-Kirchhoff material. 
  
\subsection{The Ciarlet-Geymonat material}
\label{koiter}
Ciarlet and Geymonat~\cite{CiaGey-1982} introduced a general class of hyperelastic  potentials intended to provide an extension to compressible materials of the Mooney-Rivlin stored energy (see, for example, p.\,189 of \cite{Cia-1988}). Here we shall consider a special example of this general class of materials, for which the stored elastic energy is
\begin{equation}\label{eq-CG}
W_\mathrm{CG}:=aI_1+bI_3-\frac12c\ln I_3+d,
\end{equation}
where $a>0$, $b>0$, $c>0$, and $d$ are material constants. Letting
\begin{equation}
\label{eq:definition_E_f}
\Ef:=\frac12(\Cf-\mathbf{I})
\end{equation}
denote the Green-Saint-Venant \emph{strain tensor} (see, for example, \cite[p.\,70]{gurtin:mechanics}), we show now that, for strains of sufficiently small norm $\vert\E_f\vert$, $W_\mathrm{CG}$ can be given the classical form for the stored elastic energy of isotropic materials, with Lam\'e coefficients, $\lambda$ and $\mu$, appropriately related to the material constants in \eqref{eq-CG}. 
\begin{proposition}
\label{prop-small}
In the limit of small strains the energy $W_\mathrm{CG}$ in \eqref{eq-CG} can be given the form
\begin{equation}\label{eq:W_CG_linearized}
W_\mathrm{CG}=\frac\lambda 2\tr^2\Ef+\mu\tr\Ef^2+O\left(\vert \Ef\vert^3\right),
\end{equation}
provided we set 
\begin{equation}
\label{eq:a_b_c}
\lambda=4b,\quad \mu=2a,\quad
c=2(a+b),\quad d=-(3a+b).
\end{equation}
\end{proposition}
\begin{proof}
It suffices to make use in \eqref{eq-CG} of the following equations
\begin{subequations}\label{eq:CG_identities}
\begin{align}
I_1&=3+2\tr\Ef,
\\
I_3&=1+2\tr\Ef+2\left(\tr^2\Ef-\tr\Ef^2\right)+O\left(\vert \Ef\vert^3\right).
\end{align}
\end{subequations}
\qed
\end{proof}
Here, we continue to represent the function $\phi$ as in \eqref{eq:phi_polynomial}.
However, no kinematic constraint will determine the functions  $\alpha(\x)$ and $\beta(\x)$; we need an alternative criterion, which we indentify in minimizing separately the two lowest orders in $h$ of the elastic energy integrated across the plate's thickness, for a given deformation $\y$ of the mid surfaces $S$. Hereafter, to improve clarity, the dependence on the in-plane variable $\x$ will be  omitted. 
\begin{proposition}
\label{prop-av}
Let $\phi$ be given as 
$\phi(x_3)=\alpha x_3+\beta x_3^2+\gamma x_3^3+O(x_3^3)$, with $\alpha>0$ to ensure local orientability to $\f$ in \eqref{eq:f_definition}. For $W_\mathrm{CG}$ as in \eqref{eq-CG}, the minimum energy density (per unit area) that can be attributed to $S$ is represented as 
\begin{equation}\label{eq:w_CG}
w_\mathrm{CG}:=\int_{-h}^h W_\mathrm{CG} {\rm d}x_3=h w_1+h^3 w_3 +O(h^5),
\end{equation}
where
\begin{subequations}\label{eq:w_CG_1_3}
\begin{align}
w_1&=2\left[a\tr\C+(a+b)\left(1-\ln\frac{(a+b)\det\C}{a+b\det\C}\right)-(3a+b)\right],\label{eq:w_1}
\\
w_3&=\frac13\frac{a(a+b)^2(32b\det\C+7a)}{(a+\det\C)^3}H^2+\frac53\frac{a^2(a+b)}{(a+\det\C)^2}b_1H\nonumber\\
&-\frac23\frac{a(a+b)[(a+\det\C)\tr\C+2(a+b)]}{(a+b\det\C)^2}K-\frac{1}{12}\frac{a^2}{a+\det\C}b_1^2,\label{eq:w_3}
\end{align}
and $b_1=\tr\left(\B\nablas\normal\right)$. Correspondingly, $\alpha$ and $\beta$ are determined as
\begin{align}
	\alpha&=\sqrt{\frac{a+b}{a+b\det\C}}\label{eq:alpha_min}
	\\
	\beta&= -\frac{a }{8(a+b\det\C)}b_1+ \frac{(a+b)(a-4b\det\C)}{4(a+b\det\C)^2}H\label{eq:beta_min}
\end{align}
\end{subequations}
\end{proposition}
\begin{proof}
By \eqref{eq:I_1} and \eqref{eq:I_3}, we can write
\begin{align}
I_1&=\tr\C+\alpha^2+2\alpha(b_1+2\beta) x_3 + [4\beta^2+6\alpha\gamma+2b_1(H\alpha^2+\beta)-\alpha^2K\tr\C] x_3^2+O\left(x_3^3\right)\label{eq:Pre_I_1}
  \\
I_3&=\det\C\{\alpha^2+4(\alpha\beta+H\alpha^2)x_3\nonumber\\
  &\qquad\qquad+2\left[2\beta^2+3\alpha\gamma+10 H\alpha^2\beta+(2H^2+K)\alpha^4\right] x_3^2\}+O\left(x_3^3\right).\label{eq:Pre_I_3}
\end{align}
Making use of both these equations in \eqref{eq-CG}, we readily arrive at 
\begin{equation}
\label{eq:pre_w_1}
w_1=2[a\tr\C+(a+b\det\C)\alpha^2-(a+b)\ln(\alpha^2\det\C)],
\end{equation}
which does not depend on either $\beta$ or $\gamma$ and, for given $\tr\C$ and $\det\C$, is minimized for positive $\alpha$ at the value in \eqref{eq:alpha_min}. Choosing $\alpha$ as in \eqref{eq:alpha_min}, we similarly compute
\begin{equation}
\label{eq:pre_w_3}
\begin{split}
w_3&=\frac23\bigg\{8(a+b\det\C)\beta^2+2\left[\frac{2(a+b)(4b\det\C-a)H}{a+b\det\C}+ab_1\right]\beta\\
&+\frac{a(a+b)(2b_1H-K\det\C)}{a+b\det\C}+2\frac{(a+b)^2[2(a+2b\det\C)H^2-aK]}{(a+b\det\C)^2}\bigg\},
\end{split}
\end{equation}
which is independent of $\gamma$ and is minimized for $\beta$ as in \eqref{eq:beta_min}. Inserting \eqref{eq:alpha_min} and \eqref{eq:beta_min} in \eqref{eq:pre_w_1} and \eqref{eq:pre_w_3}, respectively, we conclude the proof.
\qed
\end{proof}
\begin{remark}
Although $\gamma$ features in both $I_1$ and $I_3$ as expressed in \eqref{eq:Pre_I_1} and \eqref{eq:Pre_I_3}, it does not affect $w_1$ in \eqref{eq:pre_w_1} and neither it does	$w_3$ as long as $\alpha$ is chosen so as to minimize $w_1$. Our minimization criterion leaves $\gamma$ undetermined. To determine it, we should expand further the energy density $w_\mathrm{CG}$, so as to include terms of order $h^5$, which we renounce doing here. Both $w_1$ and $w_3$ would however remain unaffected by the value of $\gamma$.
\end{remark}	
\begin{remark}
Equation \eqref{eq:alpha_min} shows clearly how in the compressible case our method differs even more markedly from the classical Kirchhoff-Love hypothesis, as $\alpha=1$ only for $\det\C=1$. Moreover, as already occurs for incompressible materials, the bending content $w_3$ not only depends on the principal curvatures of $\surface$ via $H$ and $K$, but it also depends on the relative orientation of the eigenframes of $\B$ and $\nablas\normal$ via $b_1$.  	
\end{remark}
\begin{remark}
It is perhaps interesting to express both $w_1$ and $w_3$ in \eqref{eq:w_1} and \eqref{eq:w_3} 	in terms of the Lam\'e coefficients, $\lambda$ and $\mu$, associated with $W_\mathrm{CG}$ in the linearized limit \eqref{eq:W_CG_linearized}. By use of \eqref{eq:a_b_c}, we obtain that
\begin{subequations}\label{eq:w_1_w_3_mu_lambda}
\begin{align}
w_1&=\mu\tr\C+\frac12(2\mu+\lambda)\left(1-\ln\frac{(2\mu+\lambda)\det\C}{2\mu+\lambda\det\C}\right)-\left(3\mu+\frac{\lambda}{2}\right),\label{eq:w_1_mu_lambda}\\
w_3&=\frac13\frac{\mu(2\mu+\lambda)^2(16\lambda\det\C+7\mu)}{(2\mu+\lambda\det\C)^3}H^2+\frac53\frac{\mu^2(2\mu+\lambda)}{(2\mu+\lambda\det\C)^2}b_1H\nonumber\\
&-\frac13\frac{\mu(2\mu+\lambda)[(2\mu+\lambda\det\C)\tr\C+2(2\mu+\lambda)]}{(2\mu+\lambda\det\C)^2}K-\frac{1}{12}\frac{\mu^2}{2\mu+\lambda\det\C}b_1^2.\label{eq:w_3_mu_lambda}
\end{align}
\end{subequations}
\end{remark}
\begin{remark}
In the vanishing thickness limit introduced in Sect.~\ref{sec-gent}, we easily find that $w_1$ is minimized by $\C=\mathbf{I}_2$, so that correspondingly, again by Gauss' \emph{theorema egregium} (which requires $K=0$), $w_3$ takes the form
\begin{equation}
\label{eq:w_3_hat}
\widehat{w}_3=\frac{16}{3}\frac{\mu(\lambda+\mu)}{2\mu+\lambda}H^2.
\end{equation}	
\end{remark}
It is thus useful to consider the form acquired by $w_\mathrm{CG}$ in \eqref{eq:w_CG} when $\C$ is close to $\mathbf{I}_2$ and, correspondingly, $\nablas\normal$ is close to $\bm{0}$.
\begin{proposition}
\label{prop-alw1}
Let $\E:=\frac 12(\C-\I_2)$. The following asymptotic representations are valid for $w_1$ and $w_3$ in \eqref{eq:w_1_mu_lambda} and \eqref{eq:w_3_mu_lambda},
\begin{subequations}
\begin{align}
w_1&=\frac{2\lambda\mu}{\lambda+2\mu}\tr^2 \E+2\mu\tr\E^2+O\left(\vert \E\vert^3\right),\label{eq:w_1_asymp}
\\
w_3&=\frac{16}{3}\frac{\mu(\lambda+\mu)}{2\mu+\lambda}H^2
-\frac 43 \mu K
+O\left(\vert\E\vert\vert\nablas\normal\vert\left(\vert\nablas\normal\vert+\vert\E\vert\right)\right).\label{eq:w_3_asymp}
\end{align}
Correspondingly, $\alpha$ and $\beta$ in \eqref{eq:alpha_min} and \eqref{eq:beta_min} become 
\begin{align}
\alpha&=1-\frac{\lambda}{2\mu+\lambda}\tr\E+O\left(\vert \E\vert^2\right),\label{eq:alpha_min_asymp}
\\
\beta&=-\frac{\lambda H}{2\mu+\lambda}+O\left(\vert\nablas\normal\vert\vert\E\vert^2\right).\label{eq:beta_min_asymp}
\end{align}
\end{subequations}
\end{proposition}
\begin{proof}
To prove \eqref{eq:w_1_asymp} and \eqref{eq:w_3_asymp} it suffices to make use of the following (simple) estimates
\begin{subequations}\label{eq:estimates}
\begin{align}
\tr\C&=2(1+\tr\E),\label{eq:estimates_1}\\
\det\C&=1+2\tr\E+4(\tr^2\E-\tr\E^2)+O\left(\vert\E\vert^3\right),\label{eq:estimates_2}\\
H&=O\left(\vert\nablas\normal\vert\right),\label{eq:estimates_2_bis}\\ K&=O\left(\vert\nablas\normal\vert^2\right),\label{eq:estimates_3}\\
b_1&=2H+O\left(\vert\E\vert\vert\nablas\normal\vert\right)\label{eq:estimates_4}
\end{align}
\end{subequations}
Similarly, \eqref{eq:alpha_min_asymp} and \eqref{eq:beta_min_asymp} follow from inserting \eqref{eq:a_b_c} in \eqref{eq:alpha_min} and \eqref{eq:beta_min} and then using  again \eqref{eq:estimates}.
\qed
\end{proof}
\begin{remark}\label{rem:rigidity}
Both expressions for $\phi$ and $w_\mathrm{CG}$ provided by Prop.~\ref{prop-alw1} are precisely the same as those obtained in
Theorem~5.2 of \cite{CiaMad-2018}.\footnote{It is perhaps worth recalling that \eqref{eq:w_3_asymp} is just the same as the classical formula for the strain energy stored in a moderately bent plate comprised of a linearly isotropic elastic material, see \cite[p.\,133]{love:treatise}, where the Lam\'e coefficients, $\lambda$ and $\mu$, are replaced by Young's modulus $E$ and Poisson's ratio $\sigma$ (see, for example, \cite[p.\,126]{love:treatise}). Similarly, apart from a numerical prefactor due to a difference in scaling the plate's thickness, \eqref{eq:w_3_asymp} is also the same as equation (6.4) of \cite{friesecke:theorem}, which expresses the $\Gamma$-limit on isometries of the elastic free energy of an isotropic nonlinear material, see also footnote~\ref{foot:6.4} below.} This ensures well-posedness to the  minimum energy problem in the limit of small strains. In particular, $w_\mathrm{CG}$ with $w_1$ and $w_2$ as in \eqref{eq:w_1_asymp} and \eqref{eq:w_3_asymp} is the form appropriate to a plate of the elastic energy density envisaged in Koiter's theory for shells \cite{koiter:foundations_I,koiter:foundations_II}, in which stretching and bending energies are blended together, but are kept in a quadratic form. Equations \eqref{eq:w_1_w_3_mu_lambda} above provide instead the stretching and bending contents for a fully nonlinear theory of plates made of the Ciarlet-Geymonat material.
\end{remark}
\subsection{A variant of the Saint-Venant-Kirchhoff material}
Here, to provide a further application of the method proposed in this paper, we consider a variant of the classical Saint-Venant-Kirchhoff material studied in \cite{friesecke:theorem}.\footnote{As shown, for example, in \cite[p.\,155]{Cia-1988}, the classical Saint-Venant-Kirchhoff material is characterized by the following stored energy function
	\begin{equation*}
	\widetilde{W}_\mathrm{SVK}:=\frac{\lambda}{2}\tr^2\Ef+\mu\tr\Ef^2,
	\end{equation*}
	which has the same small-strain limit as \eqref{eq-fr} (see Prop.~\ref{prop-sfr}).} The stored energy density (per unit volume) of this material is
\begin{equation}
\label{eq-fr}
W_\mathrm{SVK}:=\frac\lambda2\tr^2\left(\sqrt{\Cf}-\I\right)
  +
  \mu \tr\left(\sqrt{\Cf}-\I\right)^2,
\end{equation}
where $\lambda$ and $\mu$ are material constants, which can be identified with the Lam\'e coefficients of this material, as shown by the following small-strain approximation to $W_\mathrm{VSK}$. 
\begin{proposition}
\label{prop-sfr}
Letting $\Ef$ be defined as in \eqref{eq:definition_E_f}, we can give $W_\mathrm{SVK}$ the same approximate form in \eqref{eq-CG}, valid for all isotropic materials,
\begin{equation}
\label{eq:W_SVK_approximate}
W_\mathrm{SVK}=\frac{\lambda}2 \tr^2\Ef+\mu\tr\Ef^2+O\left(\vert\Ef\vert^3\right).
\end{equation}
\end{proposition}
\begin{proof}
The desired conclusion follows easily from remarking that
\begin{equation}
\sqrt{\Cf}=\I+\Ef-\frac12\Ef^2+O\left(\vert\Ef\vert^3\right).
\end{equation}
\qed
\end{proof}
A rigorous method was devised in \cite{friesecke:theorem} to determine the bending content $w_3$ of a plate on all isometric embeddings $\y$ of $S$ in $\euclid$. There, $w_3$ is obtained as a $\Gamma$-limit on the class of deformations that minimize the stretching energy. It was also proved in  \cite{friesecke:theorem} that for all isotropic materials $w_3$ reads as the leading term in \eqref{eq:w_3_asymp} and the normal deformation $\phi$ has a quadratic representation with coefficients\footnote{See, in particular, the unnumbered formula on p.\,1494 of \cite{friesecke:theorem}, keeping in mind the different thickness scaling, as there $h$ amounts to our $2h$.\label{foot:6.4}}
\begin{equation}
\label{eq:alpha_beta}
\alpha=1,\quad\beta=-\frac{\lambda H}{2\mu+\lambda},
\end{equation}
in accord with the leading terms in \eqref{eq:alpha_min_asymp} and \eqref{eq:beta_min_asymp}.

For isometric embeddings $\y$, we can easily relax the polynomial approximation for $\phi$. Although this refinement makes our kinematic description more accurate, the bending content $w_3$ is not affected, as shown below for the material with stored energy density $W_\mathrm{SVK}$.
\begin{proposition}
\label{prop:phi_non_quadratic}	
Let $\y$ be such that $\C=\I_2$. Let $\phi$ in \eqref{eq:f_definition} be a function of class $\mathcal{C}^2$ in $x_3$ that obeys \eqref{eq-orient}. The minimum surface energy is
\begin{equation}
\label{eq:w_SVK}
w_\mathrm{SVK}=\int_{-h}^hW_\mathrm{SVK}\dd x_3=h^3\left(\frac{16}{3}\frac{\mu(\lambda+\mu)}{2\mu+\lambda}H^2-\frac43\mu K\right)+O(h^5)
\end{equation}
(where $K=0$, since $\y$ is an isometry), which is attained for
\begin{equation}
\label{eq:phi_min}
\phi=\frac{1}{2H}\frac{1}{2\mu+\lambda}\left(\lambda[1-\cosh(2Hx_3)]+\frac{\lambda^2\cosh(2Hh)+4\mu(\lambda+\mu)}{(2\lambda+\mu)\cosh(2Hh)}\sinh(2Hx_3) \right).
\end{equation}
\end{proposition}
\begin{proof}
$W_\mathrm{SVK}$ can readily be rewritten as
\begin{equation}\label{eq:proof_a}
W_\mathrm{SVK}=\mu\tr\Cf-(2\mu+3\lambda)\tr\sqrt{\Cf}+\frac{\lambda}{2}\tr^2\sqrt{\Cf}+\frac92\lambda+3\mu.
\end{equation} 
Since $\C=\I_2$, it easily follows from \eqref{eq:C_f_definition} and \eqref{eq:C_phi_definition} that
\begin{equation}\label{eq:proof_b}
\sqrt{\Cf}=\I_2+\frac12\phi\C_1+\frac12\left(\C_2-\frac14\C_1^2 \right)+\phi'\e_3\otimes\e_3+O\left(\phi^3\right).
\end{equation}
Moreover, since $K=0$, by use of \eqref{eq:C1_C2_curvature}, we also see that
\begin{subequations}
\begin{align}
\tr\Cf&=2+\phi'^2+4\phi H+4\phi^2H^2+O\left(\phi^3\right),\label{eq:proff_c}\\
\tr\sqrt{\Cf}&=2+\phi'+2\phi H+O\left(\phi^3\right).\label{eq:proof_d}
\end{align}
\end{subequations}
Inserting \eqref{eq:proff_c} and \eqref{eq:proof_d} into \eqref{eq:proof_a}, we arrive at
\begin{equation}\label{eq:proof_e}
W_\mathrm{SVK}=\left(\mu+\frac\lambda2\right)(\phi'-1)^2+2\lambda(\phi'-1)\phi H+2(2\mu+\lambda)\phi^2H^2+O\left(\phi^3\right).
\end{equation}
Integrating the latter expression over $[-h,h]$ we obtain a functional $F[\phi]$, whose Euler-Lagrange equation reads as
\begin{equation}
\phi''=4H^2\phi-\frac{2\lambda H}{2\mu+\lambda}.
\end{equation}
Solving this equation subject to $\phi(0)=0$, we obtain the following family of functions
\begin{equation}
\label{eq:family}
\phi_\xi(x_3)=\frac{1}{2H}\frac{\lambda}{2\mu+\lambda}[1-\cosh(2Hx_3)]+\xi\sinh(2Hx_3)
\end{equation}
in the parameter $\xi$. There is a single $\xi=\overline{\xi}$ that minimizes $F[\phi_\xi]$. Setting $\xi=\overline{\xi}$ in \eqref{eq:family}, we find \eqref{eq:phi_min}. Expanding $F[\phi_{\overline{\xi}}]$ in powers of $h$, we find \eqref{eq:w_SVK}.
\qed
\end{proof}	
\begin{remark}\label{rem:quadratic}
The quadratic approximation to $\phi_{\overline{\xi}}$ has the form $\phi_{\overline{\xi}}=\overline{\alpha}x_3+\overline{\beta}x_3^2$, where $\overline{\beta}$ is the same as $\beta$ in \eqref{eq:alpha_beta}, but  
\begin{equation}\label{eq:alpha_bar}
\overline{\alpha}=\frac{\lambda^2\cosh(2Hh)+4\mu(\lambda+\mu)}{(2\mu+\lambda)^2\cosh(2Hh)}=1-\frac{8\mu(\lambda+\mu)H^2}{(2\mu+\lambda)^2}h^2+O(h^4).
\end{equation}
\end{remark}

\section{Conclusion}\label{sec:conclusion}
We have revised the classical Kirchhoff-Love hypothesis, perhaps making it more apt to derive the blending of stretching and bending energies of a plate from the free-energy functional of three-dimensional nonlinear elasticity. In summary, we have achieved two main results: (i) we have shown that measures of stretching enter the bending energy (in addition to the expected measures of bending); (ii) we have reconciled the Kirchhoff-Love hypothesis to the rigorous $\Gamma$-convergence results on the ground where these can be compared with ours.

We have been concerned with developing a general method to obtain two-dimensional energies from three-dimensional ones and we tested it in a number of cases, thus reviving a good practice which Truesdell~\cite{truesdell:influence} lamented to be forgotten:
\begin{quote}
 ``In mathematical practice today it is, unfortunately, often forgotten that to derive basic equations is even so much a mathematician's duty as to study their properties.''
\end{quote}

Of course, there is much room (and hopes) for improvement and further extension of the proposed method.

First, the function $\phi$ introduced in \eqref{eq:f_definition} was almost invariably taken to be polynomial in $x_3$. One wonders whether $\phi$ could be chosen in a more general class of functions without jeopardizing our conclusions. The only exploration we did along these lines was in Prop.~\ref{prop:phi_non_quadratic}, but for isometric embeddings of $S$; this did not affect the bending content $w_3$, but had an effect on $\alpha$, which changed at the order $O(h^2)$, see \eqref{eq:alpha_bar}. The question is then whether we can expect that, as a rule, the bending content is not affected by letting $\phi$ vary in a wider class of functions.

Second, and more importantly, the representation of the deformation $\f$ in \eqref{eq:f_definition} is \emph{not} the most general possible. It would be interesting to replace \eqref{eq:f_definition} by 
\begin{equation}
\label{eq:ff_definition}
\f(\x,x_3)=\y(\x)+\phi(\x,x_3)\dir(\x),
\end{equation}
where the unit vector $\dir$ is a \emph{director} field on $S$, which contributes to the deformation of the whole  plate $\slab$ on the same footing as $\y$, representing the strains across the plate's thickness. Were we able to retrace our entire method starting from \eqref{eq:ff_definition} instead of \eqref{eq:f_definition}, the surface energy density $w$ resulting from a parent volume density $W$ would be a function of $\dir$ and $\nabla\dir$, as well as of $\y$ and $\nabla\y$.

Letting $\dir\cdot\normal>0$ throughout the deformed surface $\surface$, we find ourselves in the mist of the Cosserat director-theory for plates (and shells). This theory, which goes back to the pioneering works of the Cosserat brothers \cite{cosserat:theorie,cosserat:theorie_livre}, is admirably rephrased in modern terms in the book \cite{antman:nonlinear} (see, in particular, Chap.\,XIV). A full analysis of strain and equilibrium equations were first neatly developed in \cite{ericksen:exact}. In connection with this theory, the classical Kirchhoff-Love hypothesis was also used in \cite{naghdi:nonlinear}, always assuming $\dir\equiv\normal$. A more general thermodynamic treatment of one-director surfaces was presented in \cite{naghdi:theory}. As we also learn in Sect.~1.9 of \cite{ciarlet:mathematical_II}, this theory is intimately related to the Reissner-Mindlin theory of plates \cite{reissner:theory_1944,reissner:theory_1945,mindlin:influence}, which indeed allows for the normals to the mid surface in the undeformed configuration \emph{not} to remain normal to the deformed mid surface (as also illustrated in Sect.~5.2 of \cite{hughes:finite}).

All this body of knowledge suggests to take \eqref{eq:ff_definition} as a general representation of the deformation field within a plate and use it to perform a dimension-reduction of the three-dimensional stored energy to derive a genuine two-dimensional energy functional; a similar pursuit was undertaken in \cite{CiaMad-2018} (see, in particular Sect.~6.2).\footnote{There, the fields $\bm{\eta}$ and $\bm{\zeta}$, being parallel to one another, can be taken such that $\bm{\eta}(\x)=\alpha(\x)\dir(\x)$ and $\bm{\zeta}(\x)=\beta(\x)\dir(\x)$, respectively.} It remains to face the difficulties offered by assuming \eqref{eq:ff_definition} in our entire development, a task which, if not easy, might be desirable to undertake.

\appendix
\section{Cartesian formulae for $H$ and $K$}\label{sec:H_K}
Consider a deformation $\y$ of the planar surface $S$ such that $\det\C=1$, which ensures that $\y$ preserves the area of any portion of $S$. Letting $(\e_1,\e_2)$ be an orthonormal frame in the plane that contains $S$, we represent the deformation gradient $\nabla\y$ as
\begin{equation}
\label{eq:nabla_y}
\nabla\y=\av_1\otimes\e_1+\av_2\otimes\e_2,
\end{equation} 
where
\begin{equation}
\label{eq:a_1_a_2}
\av_1:=\nay\e_1\quad\text{and}\quad\av_2:=\nay\e_2.
\end{equation}
The vectors $\av_1(\x)$ and $\av_2(\x)$ are tangent to $\surface$ at the point $\y(\x)$; they need not be orthogonal to one another, but it follows from \eqref{eq:nabla_y} that 
\begin{equation}
\label{eq:det_C}
\det\C=a_1^2a_2^2-(\av_1\cdot\av_2)^2=\vert\av_1\times\av_2\vert^2=1,
\end{equation}
and so, the normal $\normal$ to $\surface$ can be represented as
\begin{equation}
\label{eq:normal}
\normal=\av_1\times\av_2.
\end{equation}
From \eqref{eq:nabla_y}, we obtain that
\begin{equation}
\label{eq:nabla_y_inverse_app}
\nay^{-1}=a_2^2\e_1\otimes\av_1+a_1^2\e_2\otimes\av_2-(\av_1\cdot\av_2)\left(\e_2\otimes\av_1+\e_1\otimes\av_2\right),
\end{equation}
so that
\begin{equation}
\begin{split}
\label{eq:curvature_tensor_app}
\nablas\normal=\nan\nay^{-1}&=a_2^2\nan\e_1\otimes\av_1+a_1^2\nan\e_2\otimes\av_2\\&-(\av_1\cdot\av_2)[\nan\e_2\otimes\av_1+\nan\e_1\otimes\av_2].
\end{split}
\end{equation}
It readily follows from \eqref{eq:normal} that
\begin{equation}
\label{eq:identity_app}
\nan\e=(\nabla\av_1)\e\times\av_2+\av_1\times(\nabla\av_2)\e,
\end{equation}
for any vector $\e$ in the space spanned by $(\e_1,\e_2)$. Combining \eqref{eq:normal} and \eqref{eq:identity_app}, we arrive at the identities
\begin{equation}
\label{eq:identities_app}
\nan\e\cdot\av_1=-\normal\cdot(\nabla\av_1)\e,\quad\nan\e\cdot\av_2=-\normal\cdot(\nabla\av_2)\e.
\end{equation}
Use of these in \eqref{eq:curvature_tensor_app} leads us to
\begin{equation}
\begin{split}
2H&=\tr\curvature\\
&=(\av_1\cdot\av_2)[\normal\cdot(\nabla\av_1)\e_2+\normal\cdot(\nabla\av_2)\e_1]-a_1^2\normal\cdot(\nabla\av_2)\e_2-a_2^2\normal\cdot(\nabla\av_1)\e_1,
\end{split}
\end{equation}
\begin{equation}
\begin{split}
	K&=\curvature\av_1\times\curvature\av_2\cdot\normal\\
	&=(\normal\cdot(\nabla\av_1)\e_1)(\normal\cdot(\nabla\av_2)\e_2)-(\normal\cdot(\nabla\av_2)\e_1)(\normal\cdot(\nabla\av_1)\e_2),
\end{split}
\end{equation}
where $\av_1$ and $\av_2$ are given by \eqref{eq:a_1_a_2}, $\normal$ by \eqref{eq:normal}, and
\begin{equation}
\nabla\av_i=(\nablatwo\y)\e_i,\quad i=1,2.
\end{equation}

\begin{acknowledgements}
The work of O.O. was supported financially by the Department of Mathematics of the University of Pavia as part of the activities funded by the Italian MIUR under the nationwide Program ``Dipartimenti di Eccellenza (2018-2022).'' E.G.V. wishes to thank Peter Palffy-Muhoray for having introduced him to the fascinating world of nematic elastomers and to the charm of Gent's material.
\end{acknowledgements}
\section*{Conflict of interest}
The authors declare that they have no conflict of interest.


\end{document}